\documentclass[11pt,a4paper]{article}
\usepackage{jheppub}

\usepackage{amssymb,amsmath,amsthm}

\newtheorem{proposition}{Proposition}
\newtheorem{definition}{Definition}
\usepackage{listings}
\usepackage{courier}

\newcommand{\id}{\mathrm{d}}
\def\vec#1{#1}
\def\w{\omega}
\def\dbox{\text{dbox}}
\def\xbox{\text{xbox}}

\author[a]{Yang Zhang}
\affiliation[a]{Niels Bohr International Academy and Discovery Center, Niels Bohr
  Institute, \\University of Copenhagen, Blegdemsvej 17, DK-2100
  Copenhagen, Denmark}
\emailAdd{zhang@nbi.dk}
\title{Integrand-Level Reduction of Loop Amplitudes by Computational
  Algebraic Geometry Methods}
\abstract{We present an algorithm for the integrand-level reduction of
multi-loop amplitudes of renormalizable field theories, based on
computational algebraic geometry. This algorithm uses (1) the Gr\"obner basis method to determine the basis
for integrand-level reduction, (2) the primary decomposition of an ideal to classify
all inequivalent solutions of unitarity cuts. The resulting basis
and cut 
solutions can be used to reconstruct the integrand from unitarity
cuts, via polynomial fitting techniques. The basis determination part
of the algorithm has been implemented in the Mathematica
package, BasisDet. The primary decomposition part can
be readily carried out by algebraic geometry softwares, with the
output of the package BasisDet. The algorithm works in both $D=4$ and
$D=4-2\epsilon$ dimensions, and we present some two and three-loop
examples of applications of this
algorithm.
}

\begin{document}
\maketitle

\section{Introduction}
The study of higher-loop amplitudes for gauge theories is important for
both theoretical and phenomenological reasons. The data analysis of
Large Hadron Collider (LHC) physics requires great accuracy of the
standard-model cross sections computation. For many channels, not only
the next-to-leading order (NLO) amplitudes, but also the
next-to-next-to-leading order (NNLO) amplitudes are important in order
to control theoretical uncertainties.   

The traditional Feynman diagram approach for amplitude calculation
becomes very complicated in the higher-loop
cases. Integration-by-parts (IBP) identities were used to reduce the number
of integrals in loop diagrams \cite{Chetyrkin:1981qh}, and efficiently
implemented in Laporta algorithm \cite{Laporta:2001dd}. The method of Gr\"obner
basis was used to express Feynman integrals as a linear combination of
master integrals \cite{Tarasov:2004ks,Smirnov:2005ky, Smirnov:2008iw}.

New methods based
on unitarity \cite{Bern:1994zx,Bern:1994cg, Britto:2004nc} decompose loop amplitudes as the product of tree
amplitudes and greatly simplify the computation. There
is a particularly convenient method: integrand-level reduction (OPP
method)  \cite{Ossola:2006us}, which
decomposes the amplitude directly at the integrand level. OPP method can be used to 
automatically and efficiently calculate one-loop amplitudes with
multiple legs \cite{Ossola:2007ax, Ellis:2007br,
  Mastrolia:2010nb, Badger:2010nx, Hirschi:2011pa, Bevilacqua:2011xh}. 
The original OPP method can be generalized to $D=4-2\epsilon$ at one loop
\cite{Giele:2008ve, Ellis:2008ir, Ellis:2011cr}. The polynomial division method was
first used in the integrand-level reduction for one-loop amplitudes in
ref. \cite{Mastrolia:2012bu}.

The generalized unitarity method also
applies to two-loop amplitudes: 
Buchbinder and Cachazo calculated two-loop planar amplitudes for
$\mathcal N=4$ super-Yang-Mills \cite{Buchbinder:2005wp}, by generalized unitarity. Gluza, Kajda and Kosower
\cite{Gluza:2010ws} used
a Gr\"obner basis to find IBP relations without
double propagators, and then
determined the master integrals for two-loop planar diagrams. The IBP
relations can also be generated by linear algebra, without using
Gr\"obner basis \cite{Schabinger:2011dz}. With
these master integrals, Kosower
and Larsen \cite{Kosower:2011ty} applied maximal unitarity method for two-loop planar diagrams to
obtain the coefficients of master integrals from the products of tree
amplitudes. Furthermore, Larsen \cite{Larsen:2012sx} applied this method to
two-loop six-point amplitudes, using multidimensional contour integrals. The two-loop double-box diagram
maximal-cut solutions can be related to Riemann surfaces, whose
geometry uniquely defines the contour integrals \cite{CaronHuot:2012ab}.  

Alternatively, using integrand-level reduction, Mastrolia and Ossola
\cite{Mastrolia:2011pr} applied
OPP-like methods to study two-loop $\mathcal N=4$ super-Yang-Mills
amplitudes. Badger, Frellesvig and Zhang \cite{Badger:2012dp} then
used the Gram-matrix method to find the integrand basis systematically
 for two-loop
amplitudes of general renormalizable theories. They calculated the double-box and
crossed-box contributions to two-loop four-point $\mathcal N=0,1,2,4$
(super)-Yang-Mills amplitudes.

It is interesting to generalize and automate the integrand-level
reduction to amplitudes with more legs and more than two loops. The
main limitations in the previous integrand-level reductions are,
\begin{enumerate}
 \item The basis for integrand-level reduction grows quickly as the
   number of loops increases. At three-loop order and the beyond,
   the Gram-matrix method becomes very complicated and the
   integrand-level basis is hard to obtain.
\item It is difficult to find and classify {\it all inequivalent} unitarity cut
  solutions for complicated diagrams. It is necessary to find all 
  cut solutions, to reconstruct the integrand. However, for diagrams
  with many legs or more than two loops, the solutions become
  complicated. And often two solutions appear to be different, are but
  actually equivalent after reparametrization \cite{CaronHuot:2012ab}.
We have to remove
  this redundancy before reconstructing the integrand. 
\end{enumerate}
These difficulties come from the complexity of the  algebraic
system of cut equations. The ideal approach to deal with these
problems is {\it computational algebraic geometry}. In this paper, we reformulate these
two problems as classic mathematical problems and solve them by
powerful mathematical tools,
\begin{enumerate}
  \item %Let $R$ be the polynomial ring of {\it irreducible scalar
    %products} (ISP) and $I$ be the ideal generated by unitarity cut equations in
    %ISPs. To find the integrand
    %basis, is equivalent to find a linear independent set $B$ in the
    %ring $R/I$, which satisfies constraints from
    %renormalizability condition. 
Integrand-level basis is equivalent to the linear basis in
the quotient ring, of polynomials in irreducible scalar products
(ISPs) modulo the cut equations. 
Then the integrand basis can be
   generated automatically using the standard Gr\"obner basis and polynomial
    reduction methods \cite[Ch.~5]{MR2290010}. 
\item The collection of all cut solutions is 
 an {\it algebraic set}. The latter can be uniquely decomposed to a
finite number of {\it affine varieties}. Each variety is an independent
 solution of the unitarity cuts, and different varieties are not
 equivalent by reparametrization. In practice, this decomposition
  is automatically done by {\it primary decomposition of an
    ideal} \cite[Ch.~1]{MR0463157}. This classifies all inequivalent unitarity
  cut solutions. Furthermore, dimension theory in algebraic
  geometry \cite[Ch.~1]{MR0463157}  can determine the number of free parameters in each solution.
\end{enumerate}

We implement the first part of our algorithm in the Mathematica package BasisDet which can automatically generate
the integrand-level basis. It also provides a list of irreducible
scalar products (ISP)'s and the
ideal $I$ generated by the cut equations. The latter
information can be directly used by computational algebraic geometry
software, like Macaulay2 \cite{M2}, to carry out the second part of
the algorithm. Once the primary decomposition is done, we get all 
inequivalent solutions of the unitarity cuts. Furthermore, for each
solution, the software will find the number of free
parameters. 

The package BasisDet has been tested for $D=4$ and $D=4-2\epsilon$
one-loop box, triangle and bubble diagrams, $D=4$ two-loop four-point double-box,
crossed-box, pentagon-triangle diagrams, $D=4$ two-loop five-point
double-box diagram, pentagon-box diagram, and $D=4-2\epsilon$ two-loop four-point
diagram. It has also been tested in two-loop level diagrams beyond
maximal unitarity, for example, $D=4$ two-loop four-point box-triangle, sunset and
double-bubble diagrams. The output bases have been verified for all these cases.

We have  also used this algorithm to calculate $D=4$ three-loop triple-box basis, and
have verified that terms inside the basis are
linearly-independent on the unitarity cuts. We also successfully
carried out a primary decomposition on this
diagram to find all the inequivalent cut solutions. 

This paper is organized as follows. In section 2, we briefly review
the known integrand-level reduction for one and two-loop diagrams. The
limitation of previous approaches is also pointed out. In
section 3, our new algorithm is presented, and its validity is
mathematically proven. Then in section 4, several examples are
presented for one, two and three-loop diagrams. Finally, our
conclusion and discussion on future directions are provided in
section 5. The manual for the package BasisDet is given in Appendix A.  

The package BasisDet and examples are included in ancillary files of the arXiv
version of this paper. The package and its future updates can also be
downloaded from the website, \url{http://www.nbi.dk/~zhang/BasisDet.html}.

\section{Review of integrand-level reduction methods}
In this section, we briefly review integrand-level reduction for
one and two-loop amplitudes. (See \cite{Ellis:2011cr} for detailed
review of the one-loop integrand reduction.) 

\subsection{One-loop integrand-level reduction}
\begin{figure}[h]
  \begin{center}
    \includegraphics[width=8cm]{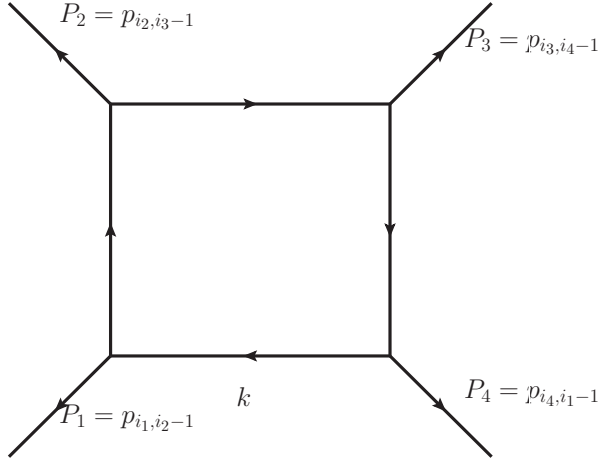}
  \end{center}
\caption{One-loop box diagram}
\label{Figure:box}
\end{figure}
Schematically, for $D=4$, an one-loop amplitude must be decomposed as \cite{Ossola:2006us},
\begin{gather}
  A_n^{(1)} =
  \int \frac{\id^4 \vec{k}}{(2 \pi)^{4/2}}
  \sum_{i_1=1}^{n-3}
  \sum_{i_2=i_1+1}^{n-2}
  \sum_{i_3=i_2+1}^{n-1}
  \sum_{i_4=i_3+1}^{n} 
  \frac{\Delta_{4,i_1i_2i_3i_4}(\vec{k})}{D_{i_1} D_{i_2} D_{i_3} D_{i_4} }\nonumber\\
  +\sum_{i_1=1}^{n-2}
  \sum_{i_2=i_1+1}^{n-1}
  \sum_{i_3=i_2+1}^{n} 
    \frac{\Delta_{3,i_1i_2i_3}(\vec{k})}{D_{i_1} D_{i_2} D_{i_3} }
  + \sum_{i_1=1}^{n-2}
  \sum_{i_2=i_1+2}^{n+i_1-2} 
    \frac{\Delta_{2,i_1i_2}(\vec{k})}{D_{i_1} D_{i_2} }
  \nonumber\\
  +\text{tadpoles, wave-function bubbles and rational terms},
  \label{eq:oneloopintegrand}
\end{gather}
where we define the propagators $D_{i_x} =
(\vec{k}-\vec{p}_{i_1,i_x-1})^2$, $\vec{p}_{i,j} =
\sum_{k=i}^j\vec{p}_k$ (figure \ref{Figure:box}) as the sum of external momenta such that
$\vec{p}_{i_1,i_1-1}=0$, and
have taken the restriction that all propagators are
massless. We must require that $\Delta_{4,i_1i_2i_3i_4}$ contain no term which is
proportional to $D_{i_x}$, $x=1,\ldots 4$, otherwise one of the
denominator in the integral is cancelled out. Similarly,
$\Delta_{3,i_1i_2i_3}$ must contain no term 
proportional to $D_{i_x}$, $x=1,\ldots 3$ and so on. 

Consider $\Delta_{4,i_1i_2i_3i_4}$ first. There exists a vector $\omega$ perpendicular to
$\vec{p}_{i_1,i_x-1}$, $x=2,3,4$. Of all the scalar products
involving loop momenta, only $k\cdot \omega$ is {\it not} a  polynomial
in denominators $D_{i_1},D_{i_2},D_{i_3}$ and $D_{i_4}$. We call such
scalar products {\it irreducible scalar products} (ISPs) and the
other scalars products {\it reducible scalar products}
(RSPs). $\Delta_{4,i_1i_2i_3i_4}$ should be a function of ISPs,
i.e. $(k\cdot \omega)$ only. 

Furthermore, we find that $\Delta_{4,i_1i_2i_3i_4}$ is at most linear
in $k\cdot \omega$,
\begin{equation}
  \Delta_{4,i_1i_2i_3i_4}=c_0+c_1 (k\cdot \omega).
  \label{eq:1lD4}
\end{equation}
$c_0$ and $c_1$ are constants independent of the loop momentum. Higher-order terms in
$(k\cdot \omega)$ are absent, because 
\begin{equation}
  (k\cdot \omega)^2 =\text{linear
combination of $D_{i_1}$, $D_{i_2}$, $D_{i_3}$ and $D_{i_4}$.}
\label{1lD4gram}
\end{equation}
The coefficients $c_0$ and $c_1$ can
be calculated from quadruple cuts, $D_{i_1}(k^{(s)})=D_{i_2}(k^{(s)})=D_{i_3}(k^{(s)})=D_{i_4}(k^{(s)})=0$,
where $s=1,2$ since there are two cut solutions. The
integrand at the two cut solutions determined the two coefficients $c_0$
and $c_1$. Note that although the term $c_1 (k\cdot \omega)$ 
integrates to zero, it is necessary to keep it for the triple-cut
calculation. We call the set $\{1,  (k\cdot \omega)\}$ the 
    integrand basis for $D=4$ one-loop quadruple cut and $(k \cdot
  \omega)$ the {\it spurious term}.

Similarly, $\Delta_{3,i_1i_2i_3}$ can be reconstructed from triple cuts. We
have two vectors $\omega_1$ and $\omega_2$, which are in perpendicular to
the external momenta. There are two ISPs, $k\cdot \omega_1$ and
$k\cdot \omega_2$. The
expansion over integrand basis reads, 
\begin{gather}
  \Delta_{3,i_1i_2i_3}(\vec{k}) = 
  c_{00}
  +c_{10}(\vec{k}\cdot \vec{\w}_1) 
  +c_{01}(\vec{k}\cdot \vec{\w}_2)
  +c_{11}(\vec{k}\cdot \vec{\w}_1) (\vec{k}\cdot \vec{\w}_2)\nonumber\\
  +c_{12}(\vec{k}\cdot \vec{\w}_1) (\vec{k}\cdot \vec{\w}_2)^2
  +c_{21}(\vec{k}\cdot \vec{\w}_1)^2 (\vec{k}\cdot \vec{\w}_2)
  +c_{20;02}\left( (\vec{k}\cdot \vec{\w}_1)^2 - (\vec{k}\cdot \vec{\w}_2)^2 \right).
  \label{eq:1lD3}
\end{gather}
The basis contains $7$ terms, of which $6$ are spurious. There are
two
cut solutions for the triple cut,
\begin{equation}
  \label{eq:18}
  D_{i_1}(k^{(s)}(\tau))=D_{i_2}(k^{(s)}(\tau))=D_{i_3}(k^{(s)}(\tau))=0,
\end{equation}
with $s=1,2$, but each of them contains one complex free parameter $\tau$.  The
original numerator of the integral at triple cuts, with all
$\Delta_{4,i_1 i_2 i_3 i_4}$ subtracted, becomes a Laurent series in
$\tau$.  The corresponding Laurent coefficients determine all the $7$
``$c$'' coefficients of eq. (\ref{eq:1lD3}).

The one-loop integrand-level reduction also works for $D=4-2\epsilon$
\cite{Ellis:2011cr}. The loop momenta contains both the four-dimensional part
and the extra-dimensional part \footnote{Throughout this paper, we use the scheme that all external momenta and polarization vectors have no $(-2\epsilon)$-dimensional components.},
\begin{equation}
  \label{eq:19}
  k=k^{[4]}+k^\perp,\quad (k^\perp)^2\equiv-\mu^2.
\end{equation}
For the quadruple cut, the basis has larger size than that of the $D=4$
case. Instead of (\ref{1lD4gram}), we have
\begin{equation}
  (k\cdot \omega)^2 -\mu^2=\text{linear
combination of $D_{i_1}$, $D_{i_2}$, $D_{i_3}$ and $D_{i_4}$.}
\label{1lD4m2egram}
\end{equation}
So we can remove either $(k\cdot \omega)^2$ or $\mu^2$ to obtain an
integrand basis. One convenient choice is 
\begin{equation}
  \Delta_{4,i_1i_2i_3i_4}^{4-2\epsilon}=c_0+c_1 (k\cdot \omega)+c_2
  \mu^2 +c_3 (k\cdot \omega) \mu^2 +c_4 \mu^4.
\label{eq:1lD42e}
\end{equation}
The $D=4-2\epsilon$ quadruple cut has only one solution. This
solution depends on one complex free parameter $\tau$. Note that
geometrically, this
solution is complex one-dimensional, and contains the two $D=4$ box
quadruple cut solutions (zero-dimensional) as two isolated points. The Taylor expansion in
$\tau$ of
the integrand, at the quadruple cut, determined the coefficients $c_0$,
$c_1$, $c_2$, $c_3$ and $c_4$. 

\subsection{Two-loop integrand-level reduction}
For the one-loop cases considered above, it is relatively easy to determine
the integrand basis and find the unitarity cut solutions. However, in
the two-loop cases, it is much harder to find the integrand basis and
the unitarity cut solutions are more complicated. 

Mastrolia and Ossola \cite{Mastrolia:2011pr}
applied the OPP-like method for two-loop $N=4$
super Yang-Mills amplitudes. Two-loop four-point amplitudes for general
renormalizable theories were calculated in \cite{Badger:2012dp}. To
show clearly new
features of two-loop integrand-level reduction, we review the
Gram-matrix method presented in ref. \cite{Badger:2012dp}. 

For example, consider the two-loop four-point planar diagram
(Figure.\ref{Figure:dbox}).  The integrand-level reduction reads,
\begin{equation}
  \label{eq:20}
   A_4^{[\dbox]}(1,2;3,4;) = 
  \int\int \frac{\id^4 \vec{k}}{(2 \pi)^{4/2}} \frac{\id^4 \vec{q}}{(2 \pi)^{4/2}}
  \Bigg(
  \frac{\Delta^{\dbox}_{7;12*34*}(\vec{k},\vec{q})}{D_1 D_2 \ldots D_7}
  \Bigg)+\ldots
\end{equation}
where $\ldots$ stands for the integrals with less than $7$
propagators. Our aim is to reconstruct the double box function
$\Delta^{\dbox}_{7;12*34*}$ from hepta-cuts (maximal cut for diagrams
with $7$ propagators). Again, there exists one vector $\omega$ which is perpendicular to all the external momenta. 
\begin{figure}[h]
  \begin{center}
    \includegraphics[width=8cm]{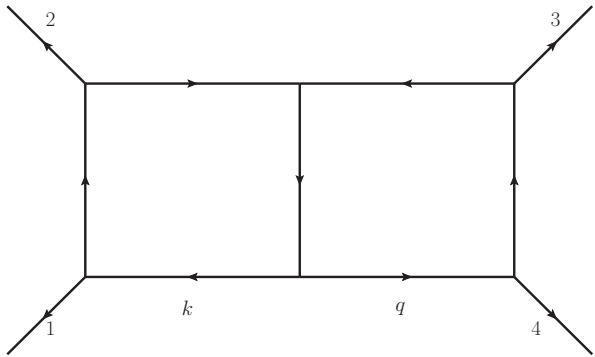}
  \end{center}
\caption{Four-point two-loop planar diagram}
\label{Figure:dbox}
\end{figure}
There are four ISPs: $(k\cdot p_4)$, $(q\cdot p_1)$, $(k\cdot
\omega)$ and  $(q\cdot \omega)$. The integrand basis consists of terms
with the form,
\begin{equation}
  \label{eq:21}
  (k\cdot p_4)^m (q\cdot p_1)^n (k\cdot \omega)^\alpha (q\cdot \omega)^\beta,
\end{equation}
where $m$, $n$, $\alpha$ and $\beta$ are non-negative integers. For
renormalizable theories, power counting requires that,
$m+n+\alpha+\beta\leq 6$, $m+\alpha \leq 4$ and $n+\beta \leq
4$. Furthermore, it is easy to see that $(k\cdot \omega)^2$, $(q \cdot
\omega)^2$ and $(k\cdot \omega)(q\cdot \omega)$ are linear
combinations of the seven denominators. Hence $\alpha\leq 1$, $\beta\leq
1$ and $\alpha\cdot \beta=0$. The above analysis is similar to that of
one-loop cases, and the integrand basis appears to contain $56$
terms. 

However, there are more constraints. For four dimension momenta, the
determinants of $5\times 5$ Gram matrices are zero,
\begin{eqnarray}
  \det G\left(
    \begin{array}{ccccc}
      1 & 2 & 4 & k & q\\
 1 & 2 & 4 & k & q
    \end{array}
\right)=0,\quad
  \det G\left(
    \begin{array}{ccccc}
      1 & 2 & 4 & k & q\\
 1 & 2 & 4 & \omega & k
    \end{array}
\right)=0,\quad
  \det G\left(
    \begin{array}{ccccc}
      1 & 2 & 4 & k & q\\
 1 & 2 & 4 & \omega & q
    \end{array}
\right)=0.
\label{Gram5-dbox}
\end{eqnarray}
For example, on the hepta-cut, the first Gram-matrix relation reads,
\begin{equation}
  0=4 (k \cdot p_4)^2 (q \cdot p_1)^2+2 s(k \cdot p_4)^2 (q \
\cdot p_1)+2 s (k \cdot p_4) (q \cdot p_1)^2-s t (k \cdot p_4) \
(q \cdot p_1) 
\label{eq:D7dboxG1}
\end{equation} 
which means that either $m\leq 2$ or $n\leq 2$.  Finally, the
integrand basis for the double box amplitude is,
\begin{align}
   \Delta^{\dbox}_{7;12*34*}(\vec{k},\vec{q}) &= 
   \sum_{mn \alpha \beta} c_{m n (\alpha+2\beta)} (k \cdot p_4)^m (q \cdot p_1)^n (k \cdot \w)^\alpha (q \cdot \w)^\beta.
\label{eq:D7dboxA}
\end{align}
There are $16$ non-spurious terms,
\begin{equation}
  (c_{000},c_{010},c_{100},c_{020},c_{110},c_{200},c_{030},c_{120},c_{210},c_{300},c_{040},c_{130},c_{310},c_{400},c_{140},c_{410}),
  \label{eq:dbox_ns}
\end{equation}
and 16 spurious terms,
\begin{equation}
  (
  c_{001},c_{011},c_{101},c_{111},c_{201},c_{211},c_{301},c_{311},
  c_{002},c_{012},c_{102},c_{022},c_{112},c_{032},c_{122},c_{132}
  ).
  \label{eq:dbox_s}
\end{equation}
We comment that, alternatively, (\ref{eq:D7dboxG1}) can be obtained
using the elimination method on the 7 cut equations. However, its computation is quite long and not systematic, comparing with the Gram-matrix method. In two-loop cases, the Gram-matrix method
provides a very efficient way to determine the basis.

Again, we can determine the $32$ coefficients from the hepta-cuts
solutions. The solutions are much more complicated
than one-loop cut solutions: there are $6$ solutions, and
each of which depends on a free parameter, $\tau$. From the Taylor or Laurent
expansion of the integrand at hepta-cuts, we can solve for the $32$
coefficients, in cases of Yang-Mills and $\mathcal N=1,2,4$
super-Yang-Mills theories. Then IBP relations can further reduce the $16$
non-spurious integrals to two master integrals. However, to get the
lower cut functions like the hexa-cut case, we have to subtract all the
$32$ terms first, not only the two master integrals. 

The four-point non-planar function
$\Delta^{\xbox}_{7;1*23*4*}$ has been determined by the same method
\cite{Badger:2012dp}. 

The Gram-matrix method becomes more complicated as we attempt to add
more loops and legs. Furthermore, it is not easy to automate the Gram-matrix method: once a Gram-matrix relation is found, we need to determine
  which monomial inside the relation should be removed from the
  basis. For diagrams with many legs or more than two loops, it is also complicated to classify all the cut
  solutions. Hence
  a new automatic algorithm is needed, to carry out integrand-level
reduction for higher-loop and many-leg amplitudes.

\section{The algorithm} 
We present an automatic algorithm for integrand-level basis determination for generalized
  unitarity, based on the techniques of computational algebraic
  geometry. The goal is (1) to find the integrand basis by Gr\"obner
  basis method (2) to classify all inequivalent unitary cut solutions by
  primary decomposition and find the dimension of each solution.

\subsection{Setup}
We parametrize the loop momenta using
scalar products. This is a variation of van Neerven-Vermaseren basis \cite{vanNeerven:1983vr}. This parameterization has the
following advantages:
\begin{itemize}
\item It does not depend on spinor helicity formalism or particular
  basis choices.   
\item The cut equations in terms of scalar products have a particularly
  simple form. This makes it convenient to carry out primary decomposition
  later. It is also easier to apply
  polynomial fitting techniques to reconstruct the integrand. 
\end{itemize}

Consider an $L$-loop $n$-point diagram. The dimension is $D=d$ or
$D=d-2\epsilon$. $d$ is an integer which stands for the dimension of
the physical spacetime, while $-2\epsilon$ is the number of extra
dimensions introduced by dimensional regularization. In
most examples, we consider $d=4$.     

Let $(l_1,\ldots l_L)$ be the
loop momenta and $(p_1, \ldots p_n)$ be the external momenta. We use
the scheme that all extra momenta and polarization vectors have no
extra-dimensional components. The momenta $p_j$
can be either massless or massive. 

We choose a basis $\{e_1,
  ..., e_d\}$ for the physical spacetime. Each $e_i$ is either an
  external momentum or an $\omega_j$, that is a momentum
  perpendicular to all the external legs. We define the
  $d\times d$ Gram matrix,
  \begin{equation}
    \label{eq:2}
    G_d=G\left(
      \begin{array}{ccc}
        e_1 , \ldots, e_d \\
 e_1 , \ldots, e_d \\
      \end{array}
\right).
  \end{equation}
$G_d$ is nonsingular, as it should be. 

For the case $D=d-2\epsilon$, we decompose the loop momenta into
physical and extra-dimensional components,
\begin{equation}
  l_i=l_i^{[d]}+l_i^{\perp},
\end{equation}
and we define $\mu_{ij}\equiv-l_i^\perp \cdot l_j^\perp$. For the case
$D=d$, we simply set $l_i^{\perp}=0$ and all the $\mu_{ij}$ are absent.

We parametrize $l_i^{[d]}$ using scalar products, 
$(l_i \cdot e_j)$, $1\leq j\leq D$,
\begin{equation}
  l_i^{[d]}=\left(e_1,
  ..., e_d\right)G_{d}^{-1}
\left(  
\begin{array}{c}
    (l_i \cdot e_1) \\
   \vdots \\
    (l_i \cdot e_d) 
  \end{array}
\right).
\label{eq:scalars-products-to-vector}
\end{equation}
We define the set of (fundamental-)scalar products (SPs) to be 
\begin{equation}
  \mathbb{SP}=\{(l_i \cdot e_j)|1 \leq i \leq L, 1\leq j\leq d\}\cup
  \{\mu_{ij}|1 \leq i \leq L, i\leq j\leq L\},\quad D=d-2\epsilon,
  \label{SPd}
\end{equation}
or 
\begin{equation}
  \mathbb{SP}=\{(l_i \cdot e_j)|1 \leq i \leq L, 1\leq j\leq d\}, \quad
  D=d.
\label{SPdm2e}
\end{equation}
All the other scalar products, like $(l_i\cdot u)$, $l_i^2$, $(l_i\cdot l_j)$,
where $u$ is a constant vector in the physical dimension, can be written as
polynomial functions of (fundamental-)scalar products, using the Gram matrix $G_d$. For
example,
\begin{eqnarray}
   l_i \cdot u=\left((l_i\cdot e_1),
  ..., (l_i \cdot e_d)\right)G_d^{-1}
\left(  
\begin{array}{c}
    (u \cdot e_1) \\
   \vdots \\
    (u \cdot e_d) 
  \end{array}
\right),\\
   l_i \cdot l_j=\left((l_i\cdot e_1),
  ..., (l_i \cdot e_d)\right)G_d^{-1}
\left(  
\begin{array}{c}
    (l_j \cdot e_1) \\
   \vdots \\
    (l_j \cdot e_d) 
  \end{array}
\right)-\mu_{ij}.
\end{eqnarray}

Next, we consider the $m$-fold unitarity cut of the amplitude, i.e., $m$
propagators are set to zero.
\begin{equation}
  \label{eq:23}
  D_1(l_1, \ldots l_L)=\ldots =D_m(l_1, \ldots l_L)=0,
\end{equation}
Using the Gram matrix $G_d$, these cut
equations can be expressed $m$ polynomial equations in the SPs. We
denote the {\it polynomial ring} of SPs, i.e. the collection of all
polynoimals in SPs, as $R'$. Then we
introduce the concept of an {\it ideal in a ring} \cite[Ch.~1]{MR2290010}: in general, an ideal
$J$ generated by
several polynomials $f_1,\ldots f_k$ in a ring $S$, is the subset of $S$,
\begin{equation}
  \label{eq:3}
  J= \langle f_1, \ldots f_k \rangle \equiv \{a_1 f_1+\ldots +a_k f_k|\forall
  a_i\in S, 1\leq i \leq k\},
\end{equation}
where $a_i$s are arbitrary polynomials in $S$. Here we define,
\begin{equation}
  \label{eq:22}
  I'=\langle D_1,\ldots D_m \rangle,
\end{equation}
which is the ideal generated by all the cut equations in terms of SPs.

Some scalar products' values are {\it uniquely} determined at {\it all} cut solutions, i.e., they are polynomials of
propagators,
\begin{equation}
  \label{eq:14}
  x = const +\mathcal O(D_1, \ldots, D_m).
\end{equation}
We may call these scalar products {\it reducible
  scalar products } (RSPs) and and all the other scalar products in
$\mathbb{SP}$ {\it
  irreducible scalar products} (ISPs). 

In practice, we can extend the definition
of RSPs. For example, if two scalar products $x_1$ and $x_2$, satisfy the
relation,
\begin{equation}
  \label{eq:15}
  \alpha_1 x_1+ \alpha_2 x_2 = const + \mathcal O(D_1, \ldots, D_m),
\end{equation}
where $\alpha_1$ and $\alpha_2$ are nonzero constant. We may pick up one of
the two scalar products as RSP, say $x_1$, and write it as a linear
function of $x_2$ on the multiple cut. 

Hence we have the following formal definition of  {\it reducible
  scalar products } (RSP) and and all the other scalar products {\it
  irreducible scalar products} (ISP):
\begin{definition}
The set $\mathbb{ISP}$ of  irreducible scalar products is a minimal subset of
$\mathbb{SP}$, such that all the other scalar products can be
expressed as linear functions in ISPs on the
unitarity cut. 
\end{definition}
This definition minimizes the number of ISPs, so the following
calculation will be simpler. The choice of $\mathbb{ISP}$ is not unique but different choices are
equivalent.  We have the following decomposition:
\begin{equation}
  \label{eq:4}
  \mathbb{SP}=\mathbb{RSP} \cup \mathbb{ISP}.
\end{equation}
%Suppose that there are $n_{R}$ reducible scalar products and $n_{I}$
%irreducible scalar products. $n_{R}+n_{I}=D\cdot L$. 
To simplify notations, we
  label the ISPs by $x_1, \dots x_{n_{I}}$. 

We can eliminate all the
  RSPs from the cut equations to obtain a new set of algebraic
  equations $\mathbb F$ in ISPs. With an abuse of notations, 
\begin{equation}
  \mathbb{F}=\{D_k(x_1, \dots x_{n_{I}})=0| 1\leq k \leq m\},
\end{equation}
where $D_k$ is the polynomial in ISPs obtained from rewriting the
$k$-th propagators in
terms of ISPs, after RSPs are eliminated. %Some $D_k$'s may be zero, and
%we keep them as trivial equations. 

We denote the polynomial ring of ISPs as $R$, and the ideal generated
by $D_\alpha$'s as $I$,
\begin{equation}
 I=\langle D_1,\ldots, D_m \rangle\subset R,
\label{Ideal:I}
\end{equation}
where $\langle ... \rangle$ stands for an ideal generated by several
polynomials.

%where the counting of equations is given by,
%\begin{equation}
 % \label{eq:5}
 % n_{\text{eqn}}=m-n_{R}=m+n_{I}-DL.
%\end{equation}
%However, the following discussion will not use the explicit result of
%counting. 

It is easy to identify the RSPs and
ISPs by hand for one and two-loop
diagrams. 
However, this calculation becomes messy for more complicated
diagrams. In practice, the identification of the RSP and ISP can be
done quickly and systematically using Gr\"obner basis method, as
described in appendix B.

The algebraic equation system $\mathbb F$ in ISPs plays the central
role in our algorithm. We will see that it contains all the
information on cut solutions and the integrand basis.

\subsection{Algorithm for integrand basis determination}
In this section, we present an automatable algorithm for the 
determination of the integrand basis. 

From the previous section, we see that all Lorentz invariants can be
reduced to polynomials of (fundamental-)scalar products. Furthermore, RSPs can
be reduced to constants or linear functions of ISPs. Hence, schematically, on $m$-fold unitarity cuts of a $L$-loop amplitude, the numerator of
the integrand is reduced to
a polynomial of ISPs, like (\ref{eq:1lD4}), (\ref{eq:1lD3}), (\ref{eq:1lD42e}) and (\ref{eq:D7dboxA}),
\begin{equation}
  \Delta_m^{L\text{-loop}}= \sum_{(\alpha_1,\ldots, \alpha_{n_I})\in S}
  c_{\alpha_1 \ldots \alpha_{n_I}} x_1^{\alpha_1} \ldots x_{n_I}^{\alpha_{n_I}}.
\label{generalized unitarity}
\end{equation}
Here the tuple $(\alpha_1, \ldots \alpha_{n_I})$ groups together the non-negative
integer powers of the
ISPs. We further require that  $\Delta_m^{L\text{-loop}}$ have no dependence in the propagators $D_i$, $i=1,\ldots,
m$.  In previous examples, this was achieved by using cut equations directly for one-loop
topologies, or Gram-matrix method for two-loop topologies.  
The finite set $S$ contains all the power tuples for
the reduction. The constant coefficients $c_{\alpha_1\ldots \alpha_{n_I}}$ are independent of loop
momenta. They can be fitted from tree amplitudes by unitarity as in
previous examples.

The set $B$ of the monomials $x_1^{\alpha_1} \ldots
x_{n_I}^{\alpha_{n_I}}$ appearing in (\ref{generalized unitarity}), is
  called the {\it integrand basis}. The goal is to determine this
  basis, or equivalently, the finite set $S$.  We translate the
  requirement that $\Delta_m^{L\text{-loop}}$ have no dependence in the propagators $D_i$, $i=1,\ldots,
m$ into mathematical language, 
\begin{proposition}
  The monomials in the integrand basis must be linearly
  independent in the quotient ring $R/I$.
\end{proposition}
\begin{proof}
  Otherwise, there exist constant coefficients $d_{\alpha_1\ldots
  \alpha_{n_I}}$, $(\alpha_1,\ldots, \alpha_{n_I})\in S$ such that 
\begin{equation}
  \sum_{(\alpha_1,\ldots, \alpha_{n_I})\in S}
  d_{\alpha_1 \ldots \alpha_{n_I}} x_1^{\alpha_1} \ldots
  x_{n_I}^{\alpha_{n_I}}=\sum_{k=1}^m f_k D_k
\end{equation}
where $f_k$'s are polynomials in ISPs. Suppose
that one coefficient $d_{\beta_1 \ldots \beta_{n_I}}$ is not zero. We
define a subset $\tilde S=S -\{(\beta_1 \ldots \beta_{n_I})\}$.
Then $\Delta_m^{L\text{-loop}}$  can be reduced even further,
\begin{gather}
\Delta_m^{L\text{-loop}}= \sum_{(\alpha_1,\ldots, \alpha_{n_I})\in S}
  c_{\alpha_1 \ldots \alpha_{n_I}} x_1^{\alpha_1} \ldots
  x_{n_I}^{\alpha_{n_I}}\nonumber \\
=\sum_{(\alpha_1,\ldots, \alpha_{n_I})\in \tilde S}
 \bigg( c_{\alpha_1 \ldots \alpha_{n_I}}-\frac{c_{\beta_1 \ldots
      \beta_{n_I}}}{d_{\beta_1 \ldots \beta_{n_I}}}d_{\alpha_1 \ldots \alpha_{n_I}} \bigg)  x_1^{\alpha_1} \ldots
  x_{n_I}^{\alpha_{n_I}} + \frac{c_{\beta_1 \ldots
      \beta_{n_I}}}{d_{\beta_1 \ldots \beta_{n_I}}} \sum_{k=1}^m f_k D_k
\label{eq: redundancy} 
\end{gather}
Thus $\Delta_m^{L\text{-loop}}$ still depends on $D_1, \ldots D_m$. We can
redefine the first term in (\ref{eq: redundancy}) as $\tilde
\Delta_m^{L\text{-loop}}$, $\tilde S$ as the new power set for the
basis, and move the second term in (\ref{eq: redundancy}) to
fewer-propagator integrals. The size of the basis decreases by one, so the reduction is not complete. 
\end{proof}

There is a classic method to find the complete linearly independent
basis in $R/I$: Buchberger's algorithm \cite[Ch.~5]{MR2290010}. (See
\cite[Ch.~2]{MR2290010} for a review of Gr\"obner basis.)
\begin{enumerate}
\item Define a
{\it monomial ordering} in $R$ and calculate the corresponding
Gr\"obner basis $G(I)$ of $I$. Denote $LT(K)$ as the collection
of leading terms of all polynomials in a set $K$, according to this
monomial ordering. 
\item Compute $LT(G(I))$, the leading terms of all polynomials in
  $G(I)$. Obtain $\langle LT(G(I)) \rangle$, the ideal generated by
  $LT(G(I))$. By the properties of Gr\"obner basis, $\langle
LT(I) \rangle=\langle LT(G(I)) \rangle$, where $\langle LT(I) \rangle$
is the ideal generated by all leading terms in $I$.
\item Then the linear basis of $R/I$ is $\hat B$, which is the set of
  all monomials which are {\it not} in  $\langle LT(I) \rangle$.
\begin{equation}
  \label{eq:7}
\hat B=  \{x_1^{\alpha_1} \ldots  x_{n_I}^{\alpha_{n_I}}  |
x_1^{\alpha_1} \ldots  x_{n_I}^{\alpha_{n_I}}\not \in \langle
LT(I) \rangle \}.
\end{equation}
\end{enumerate}

This method is fast. However, the basis generated by Buchberger's
method usually contains an infinite
number of terms, since the renormalizablity conditions have not been
imposed. We find that after the ring $R$ is reduced to $\hat B$, it is not
easy to impose renormalizablity conditions.
 
Hence, we propose the following alternative algorithm for basis determination,
based on multivariate synthetic division,
\begin{enumerate}
\item Define a
{\it monomial ordering} in $R$ and calculate the corresponding Gr\"obner basis $G(I)$ of
$I$. 
\item Generate the set $A$ of all monomials in ISPs, which satisfy
  renormalizablity conditions. $A$ must be a finite set.
\item For each monomial $a_j(x_1\ldots
  x_{n_I})$ in $A$, $1\leq j\leq |A|$,
  \begin{itemize}
    \item Carry out the multivariate synthetic division of $a_j$ by the
  Gr\"obner basis $G(I)$.
  \begin{equation}
    \label{eq:10}
    a_j(x_1\ldots
  x_{n_I})=g_j(x_1\ldots
  x_{n_I})+r_j (x_1\ldots
  x_{n_I}),\quad g_j(x_1\ldots
  x_{n_I}) \in I
  \end{equation}
where $r_j (x_1\ldots
  x_{n_I})$ is the remainder of multivariate synthetic division. Given
 the
  Gr\"obner basis $G(I)$, $r_j (x_1\ldots
  x_{n_I})$ is uniquely determined.
\item Decompose $r_j (x_1\ldots
  x_{n_I})$ as monomials and collect them in a set $B_j$.
     \end{itemize}
\item The integrand basis $B$ is then,
  \begin{equation}
    \label{eq:12}
    B=\bigcup_j B_j.
  \end{equation}
\end{enumerate}
The validity of this algorithm can be verified as follows,
\begin{itemize}
\item The monomials in $B$ are linearly independent in $R/I$. Multivariate synthetic division by
Gr\"obner basis ensures that all monomials in $B_j$ are not in $\langle
LT(I) \rangle$, therefore $B$ is a subset of $\hat B$. So by a
corollary of Buchberger's method, linear independence is proven. 
\item The basis $B$ is big enough for integrand-level reduction. From
  step 3, we see that every
  renormalizable term in the numerator of the integrand is reduced to
  monomials in $B$. In other words, it is a sum of a linear
  combination of monomials in $B$ and other terms vanishing on the
  unitarity cut.  
  
\end{itemize}
We implement this part of our algorithm in the Mathematica package BasisDet. The Gr\"obner basis
calculation and
multivariate synthetic division are done by the functions in
Mathematica. The monomial order is chosen as ``degree lexicographic'' (``deglex'' in
mathematica language, see \cite[Ch.~2]{MR2290010} for a review of monomial orderings.) and the
coefficient field can be chosen as rational functions for analytic
computation, or rational numbers for numeric computation.
%\begin{center}
%\begin{tabular}{l|c|c|c|c|c|c}
%Diagram & $\#$ ISP &  $\#$ equations &  $\#$ renormalizable terms
%&$\#$ B  &  $\#$ spurious terms & times used (sec)\\ 
%\hline
 % Two-loop double box & 4 & 3 & 160 & 32 & 16 & 0.29\\
%\hline
%Three-loop triple box & 7 & 5 & 3165 & 398 & 199 & 48\\
%\hline
%\end{tabular}
%\end{center}

\subsection{Primary decomposition of cut solutions}
Given the cut equations in ISP variables, or equivalently, the ideal $I$, the following
questions naturally arise:
\begin{itemize}
\item How many inequivalent cut solutions are there?
\item For each cut solution, how many free parameters are needed to
  parametrize it? In the other world, which is the dimension of each
  cut-solution?
\end{itemize}
These questions can be studied systematically using algebraic
geometry. We again translate these problems to mathematical
language. Consider the affine space $A^{n_I}=(x_1,\ldots
x_{n_I})$. The ideal $I$ defines an {\it affine
algebraic set}
$Y$ in $A^{n_I}$, \cite[Ch.~1]{MR0463157} 
\begin{eqnarray}
  \label{eq:26}
  Y\equiv \mathcal Z(I) =\{ (z_1,\ldots
z_{n_I}) | D_k(z_1,\ldots z_{n_I})=0,\quad \forall k\}   
\end{eqnarray}
which is the collection of all cut solutions in term of ISPs.

In general, $Y$ can always be decomposed uniquely to the union of a finite
number of irreducible components \cite[Ch.~1]{MR0463157},
\begin{equation}
  \label{eq:6}
  Y=\bigcup_{a=1}^{n_{\text{sol}}} Y_a, \quad Y_a \not \subset Y_b,\
  \text{if } a\not =b.
\end{equation}
where each $Y_a$ is an affine variety. Here, $n_{\text{sol}}$ is the number
of irreducible components. Different components are not related by parameter
redefinition. Each irreducible component corresponds to a cut
solution. So we have the following proposition:
\begin{proposition}
  The inequivalent cut solutions are in one-to-one correspondence with the irreducible components of the
algebraic set $Y$. In particular, the number of inequivalent cut
solutions equals the number of the irreducible components of the
algebraic set $Y$.
\end{proposition}
 
This decomposition can be achieved easily by the algebraic method, {\it
  primary decomposition of an ideal} \cite[Ch.~1]{MR0463157}. Since $R$ (the polynomial
ring of ISPs) is a Noetherian ring,
the {\it primary decomposition} of $I$ uniquely exists (Lasker-Noether
theorem) \cite{MR0242802},
\begin{equation}
  \label{primary-decomposition}
  I=\bigcap_{a=1}^{s} I_a.
\end{equation}
where $s$ is a finite integer and each $I_a$ is a primary ideal. Furthermore, the
primary decomposition guaranteed that,
\begin{gather}  
  \sqrt{I_a}\not =\sqrt{I_b},\quad \text{if}\quad a\not=b.\\
  I_a \not \supset \bigcap_{b\not =a} I_b, \quad \forall a.
\label{eq:8}
\end{gather}
where $\sqrt{I_a}$ is the radical of $I_a$. \footnote{The radical of an ideal
$J$ is the set of all elements $a$, such that $a^n\in J$, where $n$ is
some positive integer. \cite[Ch.~4]{MR2290010}} Because $I_a$ is primary,
$\sqrt{I_a}$ is a prime ideal. 

Hence we have the corresponding decomposition of $Y$. Define $\mathcal
Z(I_a)$
to be the zero-locus (set of all solution points) of the ideal $I_a$, 
\begin{equation}
  \label{eq:9}
  Y=\mathcal Z(I)=\bigcup_{a=1}^{s}\mathcal Z(I_a)=\bigcup_{a=1}^{s}\mathcal Z(\sqrt{I_a}).
\end{equation}
Since $\sqrt{I_a}$ is prime, $\mathcal Z(I_a)=\mathcal Z(\sqrt{I_a})$
is an affine variety \cite[Ch.~1]{MR0463157}
which is irreducible. Then we define $n_{\text{sol}}=s$ and
$Y_a=\mathcal Z(\sqrt{I_a})$, and the decomposition is done.

The dimension of each component is given by dimension theory in
commutative algebra \cite[Ch.~1]{MR0463157},
\begin{equation}
  \label{dimension-theory}
  \dim Y_a=n_I-\text{height}(\sqrt{I_a}).
\end{equation}
where $\text{height}(\sqrt{I_a})$ is the height of the prime ideal $\sqrt{I_a}$,
which is defined to be the largest integer $N$, for all possible
series of prime
ideals $\langle 0\rangle= p_0 \subset p_1 \subset \ldots \subset
p_N=\sqrt{I_a}$. Recall that $\mathbb{SP}$ is the set of
(fundamental-)scalar products, which is defined in (\ref{SPd}). Note the $\dim Y_a$ may not equal $|\mathbb{SP}|-m$, the difference
between the number of (fundamental-)scalar products and the number of cut
equations, because of the
possible redundancy in the cut equations. Furthermore, for $a
\not =b$, $\dim Y_a$ may not equal $\dim Y_b$, since they are
independent components.

Once all the irreducible components are obtained, we can parametrize each
inequivalent solution. Together with the RSPs, the explicit form for
loop momenta $l_i$ at each cut solution can be recovered. 

 The primary decomposition (\ref{primary-decomposition}) and dimension 
 (\ref{dimension-theory}) can be calculated using
 computational algebraic geometry software, for example, by 
 standard built-in functions in Macaulay2
 \cite{M2} by
 Daniel Grayson and Michael Stillman. Alternatively, if we only need
 the number of irreducible components, then a numeric algebraic geometry approach could be applied, as described in \cite{Mehta:2012wk}.

\section{Examples}
We implemented the basis determination part of our algorithm in the
Mathematica package BasisDet. The only required inputs are the kinematic relations for
external legs, a list of propagators and the renormalization
conditions.  The output is the integrand basis. It also
provides $I$, as in (\ref{Ideal:I}), the ideal generated by the cut equations in terms of
ISPs. Then we can carry out the primary decomposition and dimension
theory calculation in the computational algebraic geometry program, Macaulay2, with
the ideal $I$
obtained from BasisDet. Here we list several
examples of application of our algorithm. All computations were done on a
laptop with an Intel core i7 CPU. 

\subsection{$D=4-2\epsilon$ one-loop four-point box topology}
Take $D=4-2\epsilon$, and consider the one-loop contribution with box
topology to four-point-all-massless amplitude. The BasisDet package takes $0.05$ seconds
to generate the basis in the analytic mode (see the appendix A for the
modes of the package), 
\begin{equation}
  \Delta_{4,i_1i_2i_3i_4}^{4-2\epsilon}=c_0+c_1 (k\cdot \omega)+c_2
  \mu^2 +c_3 (k\cdot \omega) \mu^2 +c_4 \mu^4.
\end{equation}
which is exactly the same basis as (\ref{eq:1lD42e}), which was obtained
in ref. \cite{Ellis:2011cr}. The package automatically find the two ISPs $(k\cdot \omega)$
and $\mu^2$. The cut equations, after all RSP are eliminated,
become one equation,
\begin{equation}
  \label{eq:11}
  4 (k\cdot \omega)^2 -\frac{4 t u}{s} \mu^2 -t^2=0,
\end{equation}
where $s,t,u$ are Mandelstam variables. It is clear that this equation defines an irreducible parabola
in the parameter space $(k\cdot \omega,\mu^2)$. So there is only
one solution, with the dimension $1$. As a trivial test, we can also see
that from primary decomposition. Let,
\begin{equation}
  \label{eq:17}
   I=\langle 4 (k\cdot \omega)^2 -\frac{4 t u}{s} \mu_1^2 -t^2 \rangle .
\end{equation}
Macaulay2 determines that $I$ itself is primary, so no
decomposition is needed.  It also automatically finds that $\dim
I=1$, which means there is one free parameter for the solution.

\subsection{Two-loop examples}
First, we consider $D=4$ four-massless-particle amplitude with
two-loop double-box topology (Figure. \ref{Figure:dbox}). The BasisDet package takes $0.95$ seconds
to generate the basis in the analytic mode, or $0.43$ seconds to
generate the same basis in the numeric mode.
\begin{align}
   \Delta^{\dbox}_{7;12*34*}(\vec{k},\vec{q}) &= 
   \sum_{mn \alpha \beta} c_{m n (\alpha+2\beta)} (k \cdot p_4)^m (q \cdot p_1)^n (k \cdot \w)^\alpha (q \cdot \w)^\beta.
\end{align}
with $16$ non-spurious terms,
\begin{equation}
  (c_{000},c_{010},c_{100},c_{020},c_{110},c_{200},c_{030},c_{120},c_{210},c_{300},c_{040},c_{130},c_{310},c_{400},c_{140},c_{410}),
\label{BASIS:dbox_s}
\end{equation}
which are exactly the same as (\ref{eq:dbox_ns}). There are also $16$
spurious terms,
\begin{equation}
  (
  c_{001},c_{011},c_{101},c_{201},c_{301},
  c_{002},c_{012},c_{022},c_{032},c_{102},c_{112},c_{122},c_{132},c_{202},c_{302},c_{402}
  ).
\label{BASIS:dbox_ns}
\end{equation}
Note that although the number of terms is the same as (\ref{eq:dbox_s}),
some terms are different from (\ref{eq:dbox_s}). It means we
get a different but equivalent integrand basis. 
%It is a common phenomenon to have two different but equivalent
%integrand basis \cite{Ellis:2011cr} .  
Even for the one loop $D=4-2\epsilon$ box quadruple
cut, there are already several different choices of basis. We can
check explicitly
that the difference between (\ref{BASIS:dbox_ns}) and (\ref{eq:dbox_s})
is proportional to the seven propagators, so it does not change the
double-box contribution to the amplitude. 

There are four ISPs, $(l_1 \cdot p_4)$, $(l_2 \cdot p_1)$, $(l_1 \cdot
\omega)$ and $(l_2 \cdot \omega)$. The cut equations in ISPs read,
\begin{gather}
  \label{eq:1}
f_1\equiv-t^2+4 t (l_1\cdot p_4)-4 (l_1\cdot p_4)^2+4 (l_1\cdot \omega)^2=0,\\
f_2\equiv-t^2+4 t (l_2\cdot p_1)-4 (l_2\cdot p_1)^2+4 (l_2\cdot \omega)^2=0,\\
f_3\equiv s \left(-(l_1\cdot p_4)^2-2 (l_1\cdot p_4) (l_2\cdot p_1)+(l_1\cdot \omega)^2+2 (l_1\cdot \omega) (l_2\cdot \omega)-(l_2\cdot p_1)^2+(l_2\cdot \omega)^2\right)\nonumber\\-4 t (l_1\cdot p_4) (l_2\cdot p_1)=0.
\end{gather}
In this case, the ideal $I=\langle f_1, f_2, f_3 \rangle$ is quite
complicated. It is not easy to find the inequivalent solutions by hand
or by elementary analytic geometry. We use primary decomposition to find inequivalent
solutions automatically, for example, in Macaulay2, in just a couple
of seconds,
\begin{equation}
  \label{eq:27}
  I=\bigcap_{i=1}^6 I_i ,
\end{equation}
where $I_i$s are six primary ideals:
\begin{eqnarray}
  \label{eq:28}
  I_1&=&\langle l_1 \cdot p_4, 2(l_2\cdot \omega) - 2(l_2\cdot p_1) +t, 2(l_1\cdot \omega) - t \rangle \\
  I_2&=& \langle l_1 \cdot p_4, 2(l_2\cdot \omega) +2(l_2\cdot p_1) - t, 2(l_1\cdot \omega) + t \rangle \\
  I_3&=&\langle l_2\cdot p_1, 2(l_2\cdot \omega) + t, 2(l_1\cdot \omega) + 2(l_1 \cdot p_4) - t \rangle \\
  I_4&=& \langle l_2\cdot p_1, 2(l_2\cdot \omega) - t, 2(l_1\cdot \omega) - 2(l_1 \cdot p_4) + t \rangle \\
 I_5&=&\langle 2(l_2\cdot \omega) - 2(l_2\cdot p_1) + t, 2(l_1\cdot
 \omega) - 2(l_1 \cdot p_4) + t,\nonumber \\
 &&4(l_1 \cdot p_4)(l_2\cdot p_1) + 2(l_1 \cdot p_4)s +
  2(l_2\cdot p_1)s - st \rangle \\
  I_6&=& \langle 2(l_2\cdot \omega) + 2(l_2\cdot p_1) - t, 2(l_1\cdot
  \omega) + 2(l_1 \cdot p_4) - t, \nonumber \\&&4(l_1 \cdot p_4)(l_2\cdot p_1) + 2(l_1 \cdot p_4)s +
  2(l_2\cdot p_1)s - st \rangle .
\end{eqnarray}
So there are $6$ inequivalent unitarity cut solutions, consistent with
\cite{Kosower:2011ty, Mastrolia:2011pr}. Furthermore, Macaulay2
automatically finds that every solution of $I_i$ has dimension $1$. 

Note that all $I_i$'s
are generated by simple polynomials, so it is straightforward to solve
them for ISPs. Then using the Gram matrix relation,
(\ref{eq:scalars-products-to-vector}), we can rewrite the solutions in
terms of the loop momenta $l_1$ and $l_2$ and find the one-to-one
correspondence with the six solutions in ref. \cite{Kosower:2011ty}. However, this step is not
necessary since we can fit the coefficients $c_{m n (\alpha+2\beta)}$
directly from the solution for ISPs, as described in ref. \cite{Badger:2012dp}.

%\subsection{$D=4$  two-loop five-point double-box diagram}
%We consider a five-point amplitude example with two-loop double-box
%diagram. The topology is similar to the previous example, but the kinematic relations are more
%complicated. 

%The package generates the basis in about $0.57$ seconds in the numeric
%mode. Note that since there are $5$ external legs, the vector $\omega$
%does not exist. The basis contains $32$ terms, but there is no
%spurious term inside. 

Similarly, we can apply the same method on other two-loop diagrams using
BasisDet and Macaulay2. Several examples are listed in Table \ref{Two-loop-examples}.
 \begin{table}[ht!]
\centering
\begin{tabular}{|c|c|c|c|c|c|c|}
\hline
&Diagram & \#ISP & $n_{NS}$ & $n_{S}$ & $n_{basis}$ & \#Solution\\
\hline
Box-triangle &\includegraphics[scale=0.2]{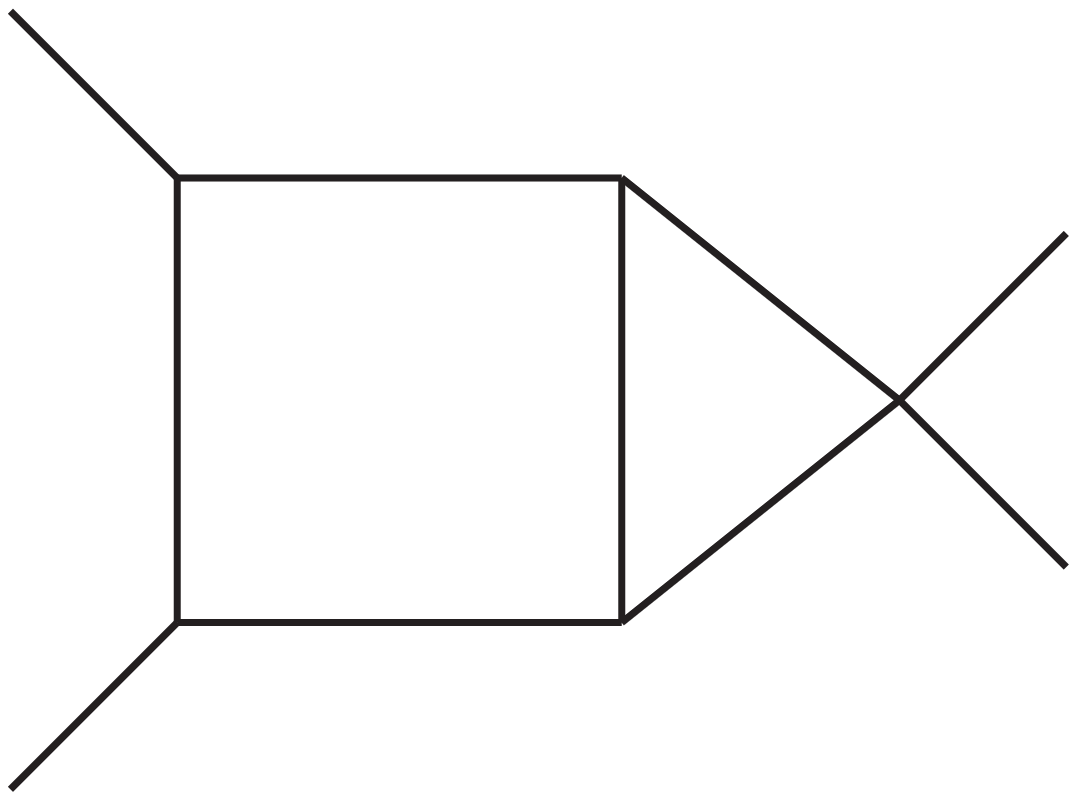}& 5 & 18 & 51 &
69 & 4\\
\hline
Five-point double-box &\includegraphics[scale=0.15]{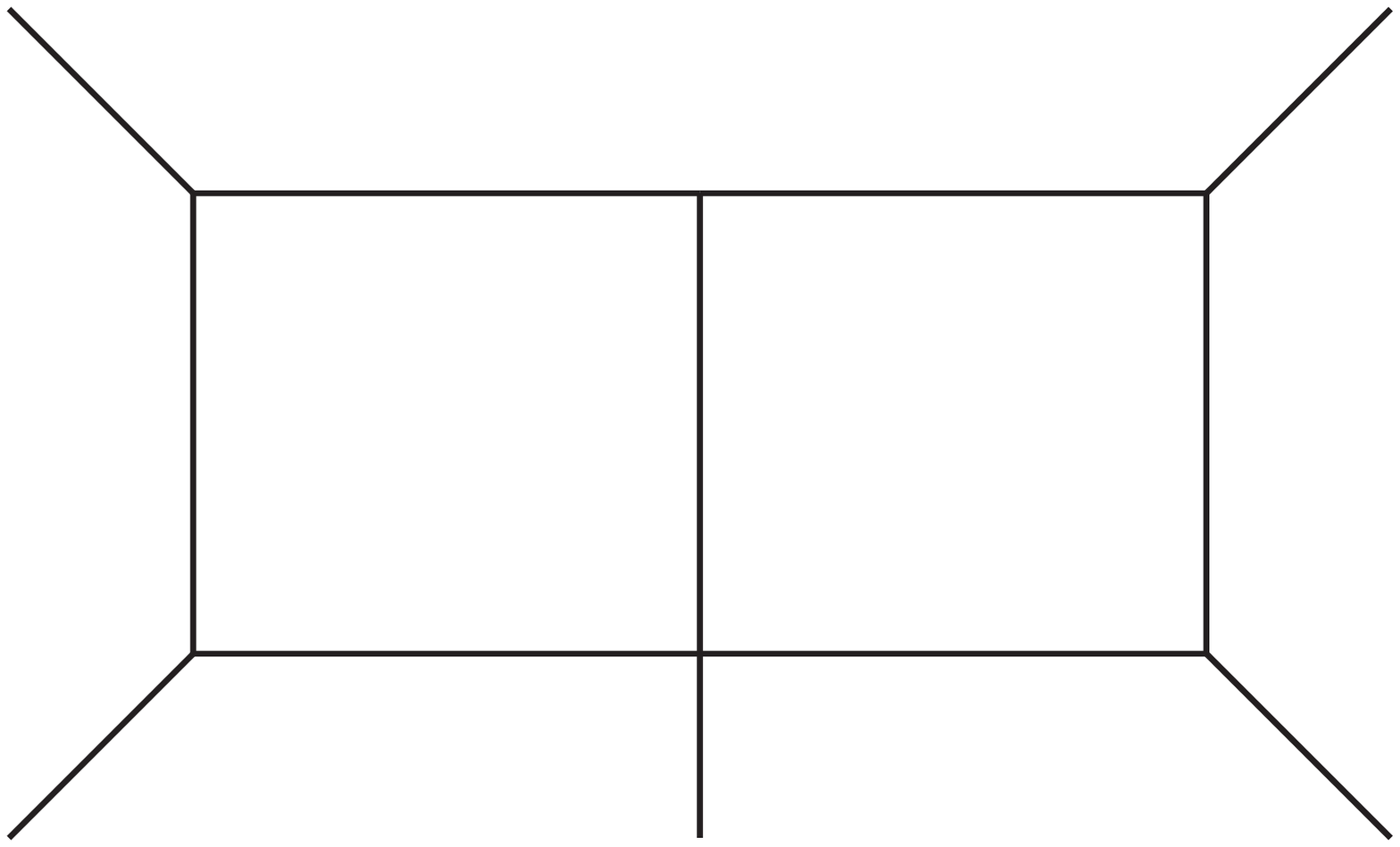}& 4 & 32 & 0 & 32 & 6\\
\hline
Sunset &\includegraphics[scale=0.15]{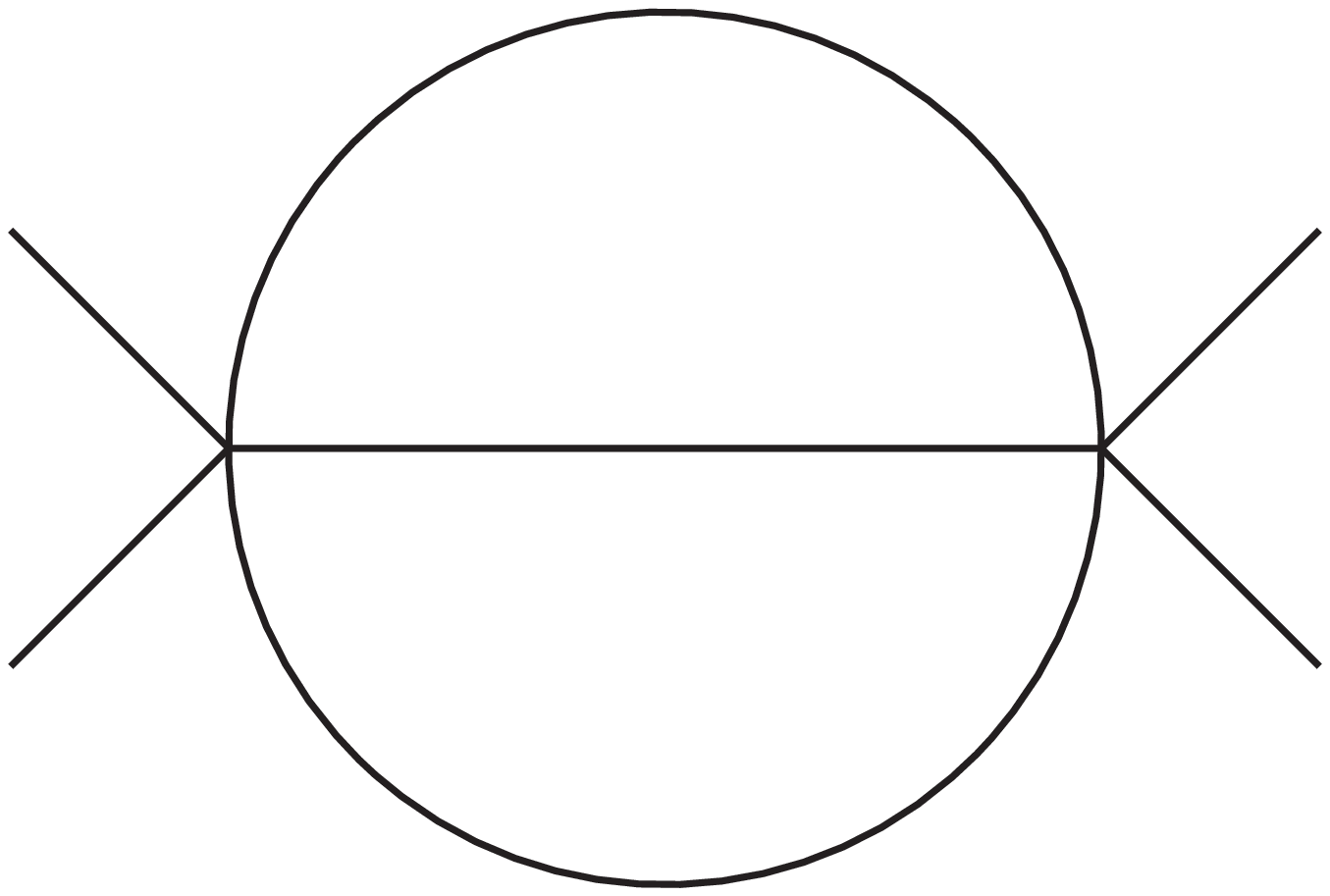}& 8 & 12 & 30 & 42 & 1\\
\hline
Double-bubble &\includegraphics[scale=0.15]{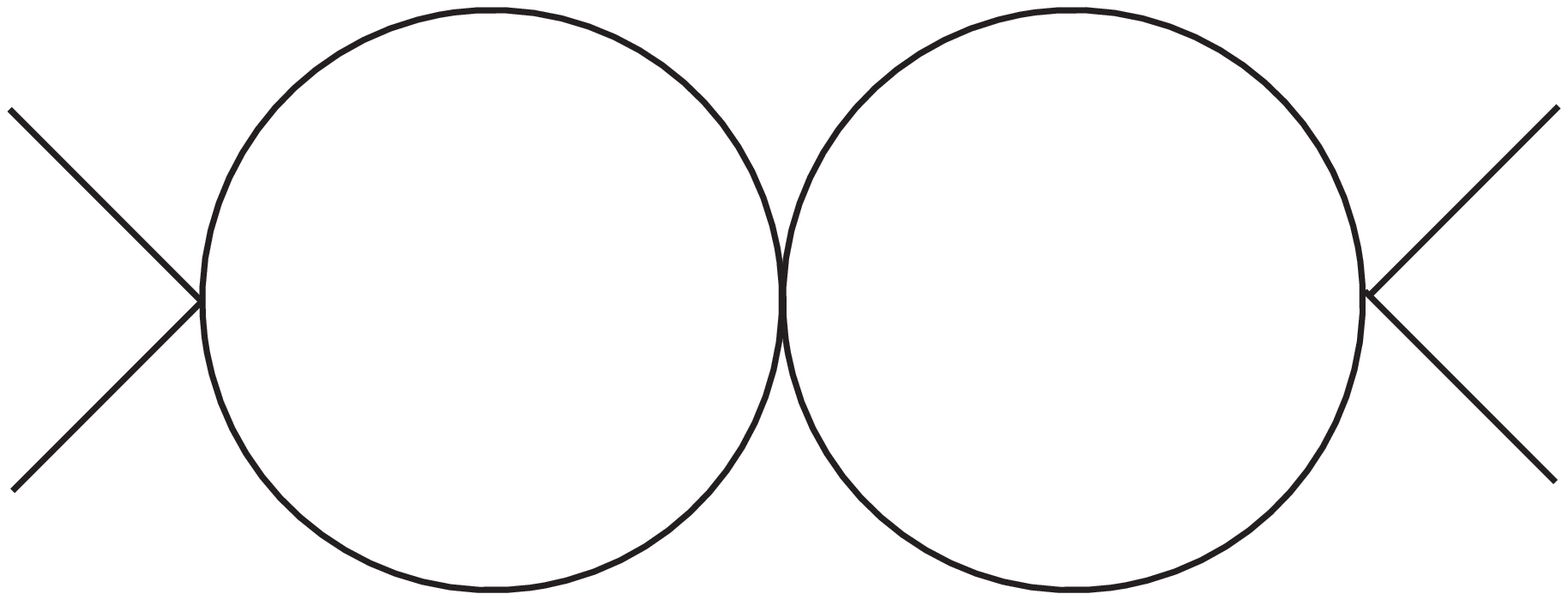}& 6 & 8 & 48 & 56 & 1\\
\hline
\end{tabular}
\caption{Several examples of the integrand-level reduction of $D=4$ two-loop diagrams: All external legs are
  massless. ``\#ISP'' is the number of ISPs. $n_{NS}$ and $n_{S}$ are the numbers of
non-spurious and  spurious terms in the integrand basis,
respectively. $n_{basis}=n_{NS}+n_{S}$ is the total number of
terms. ``\#Solution'' is the number of inequivalent solutions. The
explicit expression of the integrand basis can be obtained by running
the code ``example.nb'' with BasisDet.}
\label{Two-loop-examples}
\end{table}

\subsection{$D=4$ three-loop triple-box topology}
Consider $D=4$, four-massless-particle diagram with
three-loop triple-box topology (figure \ref{Figure:tribox}).
\begin{figure}[ht]
  \begin{center}
    \includegraphics[width=12cm]{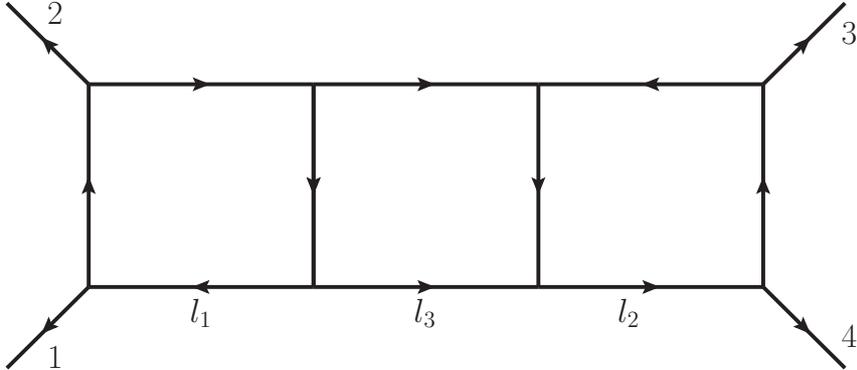}
  \end{center}
\caption{Four-point three-loop planar diagram}
\label{Figure:tribox}
\end{figure}
The
package uses about $42$ seconds in numeric mode or $4$ minutes in
analytic mode, to generate the same integrand basis. It contains
$199$ non-spurious terms and $199$ spurious terms.  

Furthermore we used Macaulay2 to find the inequivalent cut solutions by
primary decomposition. It takes about $2$ minutes to get, 
\begin{equation}
  \label{eq:29}
  I=\bigcap_{i=1}^{14} I_i,
\end{equation}
so there are $14$ inequivalent solutions. And,
\begin{equation}
  \label{eq:30}
  \dim I_i=2,\quad  i=1,\ldots,14 .
\end{equation}
Every solution thus depends on two free parameters. These solutions have been
verified both analytically and numerically. Furthermore, by the
explicit solutions, we can check
that the $398=199+199$ terms in the basis are linearly independent on
the unitarity cuts. This validates the basis.

With the integrand basis and all inequivalent solutions, we can
reconstruct the triple-box contribution to three-loop amplitude for any
renormalizable theory, via the polynomial fitting
techniques. %This work is under progress \cite{BFZ2012ii}. 

%Furthermore, we also run the package for  four-massless-particle diagram with
%four-loop quadruple-box topology. It takes about $3$ hour to finish
%the computation in the numeric mode. There are $2320$ non-spurious
%terms and $2320$ spurious terms in the basis. However, we have not
%verify the basis so far. 

\section{Conclusions and future directions}
In the paper, we have presented a new method for integrand-level
reduction, based on computational algebraic geometry. It applies (1)
a Gr\"obner basis to find the basis for integrand-level reduction, (2)
a primary decomposition of ideals to classify all inequivalent solutions
of unitarity cuts. The first part is realized in our Mathematica
package BasisDet, which automatically generates the basis from the
propagator information. This package also generates the ideal of the
cut equations, which can be used as the input of primary
decomposition. Then computational algebraic geometry software, like
Macaulay2 \cite{M2} can classify all inequivalent cut
solutions and determine the dimension of each solution.

Since this method has no dependence on the spacetime dimension, the
number of loops or the number of external legs, it works
for general multi-loop diagrams of renormalizable theories. 

We applied this method to many one-loop and two-loop
topologies. We have also used it to generate the correct basis and cut
solutions for the three-loop triple-box topology. This method
presented in this paper can be used to calculate many two-loop and
higher-loop amplitudes, via polynomial fitting techniques. 

In future, it would be interesting to work on the following directions, 
\begin{itemize}
\item Symmetries in the diagram. It is interesting to find a monomial
  ordering to keep symmetries in the diagram manifest, for the basis
  determination part of our algorithm. Then this  algorithm could be sped up considerably if all symmetries could
be made manifest.
\item Automatic parametrization of each cut solutions. We would
  like to find an automatic way of parametrizing each cut solution, after
  the dimension of each solution is obtained. It would be helpful for
  the polynomial fit process, to reconstruct the integrand from
  unitarity cuts.
\end{itemize}

\section*{Acknowledgements}
YZ would like to thank Simon Badger, Emil Bjerrum-Bohr, Hjalte Frellesvig and Valery Yundin for helpful discussions and
comments. Especially, YZ wants to express the gratitude to Simon Badger for his testing of the
package BasisDet, and his careful
reading of this paper at the draft stage. YZ also acknowledges Pierpaolo
Mastrolia for the communication with him on using Gr\"obner basis
methods, and his comments on the draft of this paper. YZ is supported by Danish
Council for Independent Research | Natural Science (FNU) grant 11-107241.

\appendix
\section{Manual for the package BasisDet 1.01}
This current code is powered by Mathematica with its embedded Gr\"obner
basis function. The latest version of the package is available on the website, \url{http://www.nbi.dk/~zhang/BasisDet.html}.
\subsection{Set up}
The main program is ``BasisDet-a-b.m'', where $a.b$ is the version
number of the package.  It should be executed as,
\begin{verbatim}
<<''/path/BasisDet-a-b.m''
\end{verbatim}
where ``/path/'' is the path for the package. 

\subsection{Input for loop diagrams}
The following variables need to be defined for the basis
determination,
\begin{itemize}
\item  \texttt{L}. It is the number of loops in the diagram.
\item \texttt{Dim}. It is the dimension of the spacetime, which should be $d$ or
  $d-2\epsilon$. $d$ is a positive integer and in most cases $d=4$.
\item \texttt{n}. It is the number of external legs.
\item \texttt{ExternalMomentaBasis}. It is a list of external momenta {\it in
  the basis  for physical spacetime}. Note that for $n$-point amplitude, because  of
  momentum conservation, we can pick up at most $n-1$ external momenta
  for the basis. In summary,
  \begin{itemize}
  \item If $n<d+1$, we need to put $n-1$ external momenta in
    ``ExternalMomentaBasis''. The program will automatically name the $d-n+1$ spurious
    vectors as $\omega_1,\ldots,\omega_{d-n+1}$.
  \item If $n\geq d+1$, we need to put $d$ external momenta in 
``ExternalMomentaBasis''.
  \end{itemize}
\item \texttt{Kinematics}. This is the list for replacement rules, from the
  kinematics. Note that only the scalar products of {\it vectors in
    the basis} need to be defined. To ensure that an kinematic
  constraints are resolved, only the independent set of
  $s_{ij}$, $s_{ij}=2 p_i \cdot p_j$, can appear in this list. For
  example, for a four-point diagram, we can only use two variables of
  the three Mandelstam variables. 
\item \texttt{numeric}. It is an optional variable for the basis
  calculation. When ``numeric'' is given and the numerical
  calculation in the GenerateBasis function is enabled, all the
  Gr\"obner basis calculation is done numerically. It will speed up
  the computation by $2\sim 5$ times. However, the numeric calculation
  has the risk of meeting kinematic singularities (like infrared limit
  and collinear limit). The numeric values should be
  rational numbers, otherwise the result depends on the floating-point
  tolerance inside the Gr\"obner basis computation. 

\item \texttt{Props}. This is the list for the propagator momenta. No specific 
  order for the propagators is necessary. The direction of the propagator
  momenta is also irrelevant.  In this version, the propagators are set
  to be {\it massless. }

\item \texttt{RenormalizationCondition}. This variable
  define the constraints from the renormalizablity condition. Each constraint on the power of the loop
 momenta is expressed as a linear inequality. For example, when
 $L=3$, the loop momenta are $l_1, l_2, l_3$ and the corresponding
 powers for the loop momenta are $\alpha_1, \alpha_2, \alpha_3$. The
 constraint 
 \begin{equation}
   \label{eq:13}
   \alpha_1+\alpha_2 \leq 6,
 \end{equation}
should be given as an item $\{\{1,1,0\},6\}$ in
``\texttt{RenormalizationCondition}''. $\{1,1,0\}$ is the list of the
coefficients of $\alpha_1, \alpha_2, \alpha_3$ and $6$ is the upper
bound.
\end{itemize}
The program will name the (fundamental-) scalar products $(l_i \cdot
e_j)$ as ``xij''.

\subsection{Computation and the output}
Once the input is given, all the basis determination computation is
done by one command GenerateBasis:
\begin{itemize}
\item \texttt{GenerateBasis[1]}. This calculates the basis analytically.
\item \texttt{GenerateBasis[0]}. This calculates the basis with numeric coefficients.
\end{itemize}
The outputs are stored in the following variables,
\begin{itemize}
\item \texttt{ISP}. This is the list for irreducible scalar products. 
\item \texttt{RSPSolution}. This is the solutions for reducible scalar products at
  the unitarity cut.
\item \texttt{CutEqnISP}. This is the list for the cut equations in terms of ISPs,
  after all RSPs are eliminated. Depending on if the numeric mode is
  enabled, the coefficient can be either numeric or analytic.
\item \texttt{Basis}. This is the list for the terms in the basis. The output
  form for each term is $(\alpha_1, \ldots \alpha_{n_I})$, where $\alpha_i$ is the power of the $i$-th
  ISP.
\item \texttt{SpuriousBasis}. This is the subset of the basis which contains
  all spurious terms.
\item \texttt{NSpuriousBasis}. This is the subset of the basis which contains
  all non-spurious terms.
\item \texttt{Integrand}. This is the integrand after the reduction, which is
  an expansion over the integrand-level basis. $cc[\alpha_1, \ldots,
  \alpha_{n_I}]$ stands for the coefficient $c_{\alpha_1, \ldots
    \alpha_{n_I}}$ for the term $x_1^{\alpha_1}\ldots
  x_{n_I}^{\alpha_{n_I}}$, as described in (\ref{generalized unitarity}).
\item \texttt{Gr}. This is the Gr\"obner basis for the ideal $I$, generated by
  cut equations in ISPs. It can be exported for other purposes.
\end{itemize}
The lists ``\texttt{ISP}'' and ``\texttt{CutEqnISP}'' can be readily used for primary
decomposition and then the dimension theory part of our algorithm, in
softwares like Macaulay2 \cite{M2}.

\subsection{Example, integrand basis for two-loop double-box diagram}
\begin{verbatim}
<< "/path/Basis-050712.m"
L=2;
Dim=4;
ExternalMomentaBasis={p1,p2,p4};
Kinematics={p1^2->0,p2^2->0,p4^2->0,p1 p2->s/2,p1 p4->t/2,
[p2 p4->-(s+t)/2,\[Omega]1^2->-t(s+t)/s};
numeric={s->11,t->3};
Props={l1-p1,l1,l1-p1-p2,l2-p3-p4,l2,l2-p4,l1+l2};
RenormalizationLoopMomenta={{1,0},{0,1},{1,1}};
RenormalizationPower={4,4,6} 
GenerateBasis[1] 
\end{verbatim}
It takes about $0.95$ second to generate the basis, with analytic
calculation. A typical output is,
\begin{verbatim}
Physical spacetime basis is {p1,p2,p4,\[Omega]1}
Number of irreducible scalar products: 4
Irreducible Scalar Products:{x14,x24,x13,x21}
Cut equations for ISP are listed in the variable 'CutEqnISP'
Possible renormalizable terms: 160
The basis contains 32 terms, which are listed in the variable 'Basis'
The explicit form of the integrand is listed in the variable 'Integrand'
Number of spurious terms: 16 , listed in the variable 'SpuriousBasis'
Number of non-spurious terms: 16, listed in the variable 'NSpuriousBasis'
Time used: 0.955934 seconds
\end{verbatim}
The we can obtain the basis information from the variables,
``CutEqnISP",  ``Basis", ``Integrand", ``SpuriousBasis" and
``NSpuriousBasis".

More examples are included in the Mathematica notebook, ``examples.nb".

\section{The algorithm of identifying the ISPs}
We have the following simple algorithm to find the ISPs, which is embedded in the package BasisDet,
\begin{itemize}
\item Calculate the Gr\"obner basis $G(I')$ for the ideal $I'$ generated by  {\it cut equations} in
terms of SPs, in the polynomial order of ``deglex''. 
\item Obtain $LT( G(I'))$, the set of the leading terms in $G(I')$. The
  linear terms in $LT(G(I'))$ are the RSPs. 
\end{itemize}

 It is easy to show that this algorithm gives the correct ISPs
 according to the definition. 
 \begin{proof}
   Suppose that this algorithm generates $\{y_1, \ldots
 , y_{n_R}\}$ as the list of the RSPs in the polynomial
 ordering, while $\{x_1, \ldots
 , x_{n_I}\}$ as the list of the ISPs in the polynomial
 ordering. First, we can prove that $y_{n_R}$ is a linear function of ISPs on the cut. $G(I')$ must
contain a linear polynomial, 
\begin{equation}
\alpha y_{n_R}+\sum_i^{n_i} \beta_i x_i + \gamma \in I'
\end{equation}
where $\alpha$, $b_i$ and $\gamma$ are constants and $\alpha\not=0$. This polynomial cannot
contain other $y_j$'s, because $y_j\succ y_{n_R}$ for $j<n_R$. Here ``$\succ$''
stands for the given monomial ordering. Thus
$y_{n_R}$ is a linear function of the ISPs on the unitarity cut.

Second, by induction, all $y_j$ are linear functions of ISPs on the cut.

Third, we can prove that the ISP set is minimal. If some $x_i$ can be
represented by a linear function of other ISPs at the cut, then
\begin{equation}
x_i -\sum_{j\not =i} \alpha_j x_j +\beta  \in I'
\end{equation}
Then the leading term of this polynomial is an ISP, say $x_k$. By the
property of Gr\"obner basis, %the
%ideal generated by leading terms in $I'$ is the same ideal generated
%by the leading terms in $G(I')$. 
$\langle LT(I') \rangle=\langle LT(G(I'))
\rangle$. Because $x_k\in LT(I')$, $x_k\in \langle LT(G(I')) \rangle$. Furthermore,
since $x_k$ has degree one, it is
generated by degree-one monomials in $LT(G(I'))$: $\{y_1, \ldots, y_{n_R}\}$. 
\begin{equation}
  \label{eq:16}
  x_k=\sum_i \gamma_i y_i
\end{equation}
while $\gamma_i$'s are constants. This contradicts the assumption of ring
structure. The ISP set is thus minimal. 
\end{proof}

\bibliographystyle{JHEP}
\bibliography{Basis-determination}
\end{document}